\documentclass[11pt,twoside]{amsart}
\usepackage{latexsym,amssymb,amsmath}
\usepackage[all]{xy}
\usepackage{tikz,pgfplots}
\usepackage{afterpage}
\usepackage{float}
\usepackage{hyperref}
\usepackage{bm}
\usepackage{blkarray}
\usepackage{booktabs}
\usepackage{enumitem}
\usepackage{caption}
\usepackage{subcaption}
\usepackage{cleveref}
\usepackage{multirow}
\usepackage{soul}

\numberwithin{equation}{section}
\newtheorem{theorem}{Theorem}[section]

\newtheorem{proposition}[theorem]{Proposition}
\newtheorem{corollary}[theorem]{Corollary}

\theoremstyle{definition}

\newtheorem{remark}[theorem]{Remark}
\newtheorem{example}[theorem]{Example}

\newcommand{\lii}{[\![}

\newcommand{\rii}{]\!]}

\begin{document}

\title{On the generalized Hamming weights of hyperbolic codes}
\author{Eduardo Camps-Moreno}
\address[Eduardo Camps-Moreno]{Escuela Superior de F\'isica y Matem\'aticas \\ Instituto Polit\'ecnico Nacional\\ Mexico City, Mexico}
\email{camps@esfm.ipn.mx}
\author{Ignacio Garc\'ia-Marco}
\address[Ignacio Garc\'ia-Marco]{Departamento de Matem\'aticas, Estad\'istica e I.O. \\ Universidad de La Laguna \\ La Laguna, Tenerife, Spain}
\email{iggarcia@ull.edu.es}
\author{Hiram H. L\'opez}
\address[Hiram H. L\'opez]{Department of Mathematics and Statistics\\ Cleveland State University\\
Cleveland, OH USA}
\email{h.lopezvaldez@csuohio.edu}
\author{Irene M\'arquez-Corbella}
\address[Irene M\'arquez-Corbella]{Departamento de Matem\'aticas, Estad\'istica e I.O. \\ Universidad de La Laguna \\ La Laguna, Tenerife, Spain}
\email{imarquec@ull.edu.es}
\author{Edgar Mart\'inez-Moro}
\address[Edgar Mart\'inez-Moro]{Institute of Mathematics\\ University of Valladolid\\ Spain}
\email{edgar.martinez@uva.es}
\author{Eliseo Sarmiento}
\address[Eliseo Sarmiento]{Escuela Superior de F\'isica y Matem\'aticas \\ Instituto Polit\'ecnico Nacional\\ Mexico City, Mexico}
\email{esarmiento@ipn.mx}
\thanks{The first and sixth authors were partially supported by SIP-IPN, project 20211838, and CONACyT. The third author was partially supported by NSF DMS-2201094. The second and fourth authors were partially supported by the Spanish MICINN PID2019-105896GB-I00 and MASCA (ULL Research Project). The fourth and fifth authors were supported in part by Grant TED2021-130358B-I00 funded by
MCIN/AEI/10.13039/501100011033 and by the ``European Union
NextGenerationEU/PRTR''. (Corresponding author: Hiram H. L\'opez.)}
\begin{abstract}
A hyperbolic code is an evaluation code that improves a Reed-Muller because the dimension increases while the minimum distance is not penalized. We give the necessary and sufficient conditions, based on the basic parameters of the Reed-Muller, to determine whether a Reed-Muller coincides with a hyperbolic code. Given a hyperbolic code, we find the largest Reed-Muller containing the hyperbolic code and the smallest Reed-Muller in the hyperbolic code. We then prove that similarly to Reed-Muller and Cartesian codes, the $r$-th generalized Hamming weight and the $r$-th footprint of the hyperbolic code coincide. Unlike Reed-Muller and Cartesian, determining the $r$-th footprint of a hyperbolic code is still an open problem. We give upper and lower bounds for the $r$-th footprint of a hyperbolic code that, sometimes, are sharp.
\end{abstract}
\dedicatory{Dedicated to Joachim Rosenthal on the occasion of his sixtieth birthday.\\
We thank Prof. Rosenthal for his selfless, continuous, and endless \\
support to shape the coding theory and cryptography community.}
\maketitle
\markleft{E. Camps-Moreno, I. Garc\'ia-Marco, H. H. L\'opez, I. M\'arquez-Corbella, E. Mart\'inez-Moro, and E. Sarmiento}

\section{Introduction}
Let $\mathbb F_q$ be a finite field with $q$ elements, where $q$ is a power of a prime.
An $[n,k,\delta]$ {\it linear code} $\mathcal C$ over $\mathbb{F}_q$ is a subspace $\mathcal{C} \subseteq \mathbb{F}_q^n$ with $\mathbb{F}_q$-dimension $k$ and minimum distance $\delta := \min \{ d_H(\mathbf{c},\mathbf{c}') \ : \ \mathbf{c},\mathbf{c}' \in \mathcal C, \mathbf{c} \neq \mathbf{c}' \},$ where $d_H(\cdot,\cdot)$ denotes the {\it Hamming distance}.

The Generalized Hamming weights (GHWs) for linear codes, a natural generalization of the minimum distance, were introduced by Wei in 1992~\cite{Wey}. Wei showed in the same work~\cite{Wey} that the GHWs completely characterize the performance of a linear code when used on the wire-tap channel of type II. The GHWs are also related to resilient functions and trellis, or branch complexity of linear codes~\cite{TsVl}. The precise definition is the following. For a nonnegative integer $s$, we set $[s] := \{1, 2, ..., s\}$.
The support of a subspace $\mathcal D \subseteq \mathbb F_q^n$ is defined by
$ \chi(\mathcal D) := \left\{i \in [n] \ : \ \text{there is } \mathbf (x_1,\ldots, x_n) \text{ in } \mathcal D \text{ with } x_i \neq 0\right\}.$
For an integer $1 \leq r \leq k$, the $r$-th {\it generalized Hamming weight} of $\mathcal C$ is given by
\[\delta_r(\mathcal C) := \min \{ \ |\chi(\mathcal D)| \ : \ \mathcal D \subseteq \mathcal C, \dim (\mathcal D) = r \}.\]
Note that $\delta_1(\mathcal C)$ is the minimum distance of $C$.

This work will focus on evaluation codes whose evaluation points are the points in $\mathcal P:= \mathbb F_q^m$.
For $A\subseteq \mathbb N^m$, let
$\mathbb F_q[A]$ be the subspace of polynomials in $\mathbb F_q[\mathbf X] := \mathbb F_q[X_1, \ldots, X_m]$ with $\mathbb F_q$-basis 
$\left\{ \mathbf{X}^{\mathbf i} := X_1^{i_1} \cdots X_m^{i_m} \ : \ \mathbf i = (i_1, \ldots, i_m) \in A\right\}.$
Write $ \mathcal P = \{P_1, \ldots, P_n\}$, where $n := |\mathcal P| = q^m.$ Define the following evaluation map
\[\begin{array}{cccc}
\mathrm{ev}_{\mathcal P}: & \mathbb F_q[X_1, \ldots, X_m] & \longrightarrow & \mathbb F_q^n\\
& f & \longmapsto & (f(P_1), \ldots, f(P_n)).
\end{array}\]
The {\it evaluation or monomial code} associated with $A$ is denoted and defined by
\[\mathcal C_A := \mathrm{ev}_{\mathcal P} (\mathbb F_q[A]) = \left\{\mathrm{ev}_{\mathcal P}(f) \ : \ f\in \mathbb F_q[A] \right\}.\]
For $a, b \in \mathbb R$ and $a \leq b$, we denote by $\lii a, b \rii$ the integer interval $[a,b] \cap \mathbb Z$.
Recall $A\subseteq \mathbb N^m.$ As $\alpha^q=\alpha$ for every $\alpha\in \mathbb F_q$, one can find a unique set $B\subseteq \lii 0, q-1\rii ^m$ such that $\mathcal C_A = \mathcal C_B$. In what follows, if a set $A\subseteq \mathbb N^m$ defines the code $\mathcal C_A$, we are assuming that $A \subseteq \lii 0, q-1\rii ^m$.

Observe that the length and dimension of the evaluation code $\mathcal C_A$ are $q^m$ and $|A|$, respectively. The minimum distance of $\mathcal C_A$ has been studied in terms of the footprint that we now define.
The \emph{footprint-bound} of the evaluation code $\mathcal C_A$ is the integer
\[ \mathrm{FB}(\mathcal C_A) :=  \min \left\{ \prod_{j=1}^m (q-i_j) \ : \ (i_1, \ldots, i_m) \in A \right\}.\]
The footprint-bound matters because the minimum distance $\delta_1(\mathcal C_A)$ of $\mathcal C_A$ is lower bounded by the footprint bound~\cite{GH00}:
$\mathrm{FB}(\mathcal C_A) \leq \delta_1(\mathcal C_A)$.
The footprint-bound has been extensively studied in the literature. See, for example, \cite{BeelenP, Car18, G08b, JVV} and the references therein. 

The families of Reed-Muller and hyperbolic codes that we describe below are particular cases of evaluation codes. Let $s \geq 0, m \geq 1$ be integers and take
\[R := \left\{\mathbf i = (i_1, \ldots, i_m) \in \lii 0,q-1 \rii^m \ : \ \sum_{j=1}^m i_j \leq s \right\}.\] The evaluation code $\mathcal C_{R}$, denoted by $\mathrm{RM}_q(s,m)$, is called \emph{Reed-Muller code} over $\mathbb{F}_q$ of degree $s$ with $m$ variables.

A hyperbolic code, introduced by Geil and H\o holdt in \cite{GH01}, is an evaluation code designed to improve the dimension of a Reed-Muller code while the minimum distance is not penalized. The precise definition of a hyperbolic code is the following. Let $d, m \geq 1$ be integers and take
\[H := \left\{ \mathbf i = (i_1, \ldots, i_m) \in \lii 0,q-1\rii^m \ : \ \prod_{j=1}^m (q-i_j) \geq d \right\}.\]
The evaluation code $\mathcal C_H$, denoted by $\mathrm{Hyp}_q(d,m)$, is called the \emph{hyperbolic code} over $\mathbb{F}_q$ of order $d$ with $m$ variables.

Note that the hyperbolic code $\mathrm{Hyp}_q(d,m)$ has been designed to be the code with the largest possible dimension among those monomial codes $\mathcal C_A$ such that $\mathrm{FB}(\mathcal C_A)\geq d$. There are instances where the hyperbolic codes improve the Reed-Muller codes, meaning that the dimension has increased~\cite{HLP}. But sometimes, the hyperbolic and Reed-Muller codes coincide. In this paper, we give the necessary and sufficient conditions to determine whether a Reed-Muller code is hyperbolic; those conditions are provided in terms of the basic parameters of the Reed-Muller code. Given a hyperbolic code, we find the largest (respectively smallest) Reed-Muller code contained in (respectively that contains) the hyperbolic.

The GHWs have been studied for many well-known families of codes. Heijnen and Pellikaan introduced in~\cite{HeiPel}, in a general setting, the order bound on GHWs of codes on varieties to compute the GHWs of Reed-Muller codes. Beelen and Datta used a similar approach of the order bound in~\cite{BEELEN2018130} to calculate the GHWs of Cartesian codes. Jaramillo et al. introduced in~\cite{JVV} the $r$-th footprint to bound the GHWs for any evaluation code. This paper proves that the $r$-th generalized Hamming weight and the $r$-th footprint of a hyperbolic code coincide.

The outline of this paper is as follows. In Section~\ref{secC1}, we determine when a Reed-Muller code is hyperbolic. Thus, we indicate when the hyperbolic code of order $d$ has a greater dimension concerning a Reed-Muller code with the same minimum distance. 
Given a hyperbolic code $\mathrm{Hyp}_q(d,m),$ we find in Section~\ref{secC2} the smallest Reed-Muller code $\mathrm{RM}_q(s^\prime,m)$ that contains $\mathrm{Hyp}_q(d,m).$ Section~\ref{secC4} finds the largest Reed-Muller code $\mathrm{RM}_q(s,m)$ contained in $\mathrm{Hyp}_q(d,m).$ In other words, Section~\ref{secC2} and~\ref{secC4} find the largest $s$ and the smallest $s^\prime$ such that
\[ \mathrm{RM}_q(s,m) \subseteq \mathrm{Hyp}_q(d,m) \subseteq \mathrm{RM}_q(s^\prime,m).\]

In Section~\ref{21.04.04}, we prove that similarly to Reed-Muller and Cartesian codes, the $r$-th generalized Hamming weight and the $r$-th footprint of the hyperbolic code coincide. Unlike Reed-Muller and Cartesian, determining the $r$-th footprint of a hyperbolic code is still an open problem. We use the results from Sections~\ref{secC1},~\ref{secC2}, and~\ref{secC4} to provide upper and lower bounds for the $r$-th footprint of a hyperbolic code that, sometimes, are sharp.

\section{When hyperbolic and Reed-Muller codes coincide}\label{secC1}
We determine in this section when a Reed-Muller code is hyperbolic. In other words, for a Reed-Muller code $\mathrm{RM}_q(s,m),$ we give necessary and sufficient conditions over $q$, $s,$ $m$ and its minimum distance to determine if $\mathrm{RM}_q(s,m)$ is a hyperbolic code.

\begin{proposition}\label{20.08.10}
Assume $s=mt+r$, where $t,r \in \mathbb{N}$ and $0\leq r\leq m-1.$ Then
\[\max\left\{\prod_{j=1}^m (q-i_j)  \ : \  \sum_{j=1}^m i_j=s,\ 0\leq i_j\leq q-1\right\}=(q-t-1)^r(q-t)^{m-r}.\]
\end{proposition}
\begin{proof}
Consider $\mathbf{i}=(i_1,\ldots,i_m)$ such that $\prod_{j=1}^m (q-i_j)$ reaches the maximum value. If all the $i_j$'s are equal, then $i_j=\frac{s}{m}$ and we have the result ($r=0$). If they are not equal, we can assume by symmetry that $i_1>i_2$, and then we would have that
$$(q-i_1+1)(q-i_2-1)\prod_{j=3}^{m}(q-i_j)-\prod_{j=1}^m (q-i_j)>0$$

\noindent if and only if $i_1-i_2-1>0$. Since we have chosen $\mathbf{i}$ to be maximum, then $i_1-i_2-1=0$ and therefore $i_1=i_2+1$. This means that $i_1=\ldots=i_r=t+1$ and $i_{r+1}=\ldots=i_m=t$ for some $r$ and $t$ (and then
$s=mt+r$) and thus the conclusion follows.
\end{proof}
We come to one of the main results of this section.
\begin{theorem}\label{20.08.16}
Let $m \geq 1$ and $0 \leq s < m(q-1)$. The Reed-Muller code $\mathrm{RM}_q(s,m)$ with minimum distance $\delta$ is a hyperbolic code if and only if
\[(q-t-1)^r(q-t)^{m-r}<\delta,\]
where $s+1=mt+r$ and $0\leq r < m.$ Even more, in this case we have $\mathrm{RM}_q(s,m)=\mathrm{Hyp}_q(\delta,m).$
\end{theorem}
\begin{proof}
Define the sets
\[R=\left\{ {\bf i}   \in \lii 0,q-1 \rii^m \ : \ \sum_{j=1}^m i_j\leq s \right\} \text{ and }
H=\left\{ {\bf i}  \in \lii 0,q-1 \rii^m  \ : \ \prod_{j=1}^m (q-i_j)\geq \delta \right\}.\]
As the minimum distance of the Reed-Muller code $\mathrm{RM}_q(s,m)$ is $\delta$ \cite[Theorem 3.9]{CLMS}, for every vector ${\bf i}=\left(i_1,\ldots,i_m \right)$ such that $\sum_{j=1}^m i_j\leq s$ we have that $\prod_{j=1}^m (q-i_j)\geq \delta.$ This implies that $R\subseteq H.$ Thus, by definition of hyperbolic code, the Reed-Muller code $\mathrm{RM}_q(s,m)$ is a hyperbolic code if and only $H\subseteq R.$ Define ${\bf i}=\left(i_1,\ldots,i_m \right)$ such that $\sum_{j=1}^m i_j \geq s+1=mt+r.$ By Proposition~\ref{20.08.10}, $\prod_{j=1}^m (q-i_j) \leq (q-t-1)^r(q-t)^{m-r}.$ We conclude $R=H$ if and only if $(q-t-1)^r(q-t)^{m-r}<\delta.$ In this case, we see that  $\mathcal C_H$ is the hyperbolic code $\mathrm{Hyp}_q(\delta,m).$
\end{proof}

Theorem~\ref{20.08.16} was previously proved in \cite[Proposition 2]{GMR} for the case when $m=2$. Even when $m=2$, we can observe that, in most nontrivial cases, the hyperbolic code outperforms the corresponding Reed-Muller code with the same minimum distance.
\begin{example}
Take $q = 9$. By Theorem~\ref{20.08.16}, we have the following inequality for the dimensions of the Reed-Muller and the hyperbolic codes
$$\dim (\mathrm{RM}_9(s,2)) < \dim (\mathrm{Hyp}_9(\delta,2))$$
for all $s \in \lii 5,13 \rii$. The dimensions are equal for $s \in \lii 0,4 \rii \cup \lii 14, 16 \rii$.
\end{example}

\begin{corollary}
    For $q=2$, Reed-Muller codes and hyperbolic codes are the same.
\end{corollary}
\begin{proof}
Take $s+1 < m(q-1).$ Then $s+1=m(0)+r,$ with $0\leq r\leq m-1$. We know that
\[2^{m-r}<2^{m-r+1}=2^{m-s}=\delta(\mathrm{RM}_q(s,m)).\]
By Theorem~\ref{20.08.16} we have $\mathrm{RM}_q(s,m)=\mathrm{Hyp}_q(2^{m-s},m).$
\end{proof}

\section{The smallest Reed-Muller code}\label{secC2}
Given the hyperbolic code of degree $d$ with $m$ variables $\mathrm{Hyp}_q(d,m)$,
we now find the smallest degree $s$ such that $\mathrm{Hyp}_q(d,m) \subseteq \mathrm{RM}_q(s,m)$. We will use the following notation. The symbol $\left \lfloor a \right \rfloor$ denotes the integer part of the real number $a,$ which is the nearest and smaller integer of $a,$ and $\{a\}$ is the fractional part of $a,$ defined by the formula $\{a\} = a - \left \lfloor a \right \rfloor.$ 

We start with $m = 2,$ the case of two variables.
\begin{proposition}
\label{Prop:smallerRM-m=2}
Given $d\in \mathbb N,$ define $a=q-\sqrt{d}$ and $s=\left \lfloor 2a \right \rfloor.$ Then $\mathrm{Hyp}_q(d,2) \subseteq \mathrm{RM}_q(s,2)$. Moreover, $s$ is the smallest integer with this property, that is
\[\mathrm{Hyp}_q(d,2) \not\subseteq \mathrm{RM}_q(s-1,2).\]
\end{proposition}

\begin{proof}
Let $H, R_1, R_2 \subseteq \lii 0,q-1 \rii^m$ be the sets  defining the codes $\mathrm{Hyp}_q(d,2)$, $\mathrm{RM}_q(s,2)$ and $\mathrm{RM}_q(s-1,2)$, respectively.
We show first that $\mathrm{Hyp}_q(d,2) \subseteq \mathrm{RM}_q(s,2)$. We will use the following fact:
\begin{equation}
\min \{a_1+a_2 \ : \ a_1,a_2 \in \mathbb{R}_{\geq 0}, a_1 a_2 = d \} = 2\sqrt{d}.
\label{eq:optimization-Problem}
\end{equation}
For every $\mathbf{i}=({i}_1,{i}_2)\in \mathbb N^2$ such that $\mathbf{i}\in H,$ we have that $(q-{i}_1)(q-{i}_2) \geq d$. By Eq. \eqref{eq:optimization-Problem},
$(q-{i}_1)+(q-{i}_2) \geq 2\sqrt{d}$, i.e. ${i}_1+{i}_2 \leq 2a$. Moreover,  since ${i}_1,{i}_2\in \mathbb N$, then ${i}_1+{i}_2\leq \left \lfloor 2 a\right \rfloor = s$. Thus $\mathbf{i} \in R_1,$ which proves the first statement.

We show now that $\mathrm{Hyp}_q(d,2) \not\subseteq \mathrm{RM}(s-1,2)$. We separate it into two cases.

Case $0\leq \{a\} < \frac{1}{2}$. Take $\mathbf a = (\left \lfloor a \right \rfloor, \left \lfloor a \right \rfloor)\in \mathbb N^2$. As $(q-\left \lfloor a \right \rfloor)^2 \geq (q-a)^2 = d,$ ${\mathbf a}$ belongs to $H$. Observe that $\left \lfloor a \right \rfloor + \left \lfloor a \right \rfloor = 2\left \lfloor a \right \rfloor = \left \lfloor 2a \right \rfloor = s > s-1,$ thus $\mathbf a \notin R_2.$
    \item Case $\frac{1}{2}\leq \{ a \} < 1$. Take $\mathbf a = (\left \lfloor a \right \rfloor, \left \lfloor a \right \rfloor + 1)\in \mathbb N^2.$ Since $\{a\} \geq \frac{1}{2}$, then $q-a = q - \left \lfloor a \right \rfloor - \{a\} \leq q-\left \lfloor a \right \rfloor -\frac{1}{2}$. Thus, the equation $ \left(q-\left \lfloor a \right \rfloor - \frac{1}{2}\right)^2 = (q-\left \lfloor a\right \rfloor)(q-\left \lfloor a \right \rfloor-1) + \frac{1}{4},$ implies that $ (q-\left \lfloor a\right \rfloor)(q-\left \lfloor a \right \rfloor-1) = \left \lfloor \left(q-\left \lfloor a \right \rfloor - \frac{1}{2}\right)^2 \right \rfloor \geq \left \lfloor (q-a)^2\right \rfloor = \left \lfloor d \right \rfloor = d.$ Which means that $\mathbf a$ belongs to $H.$ As $\left \lfloor a \right \rfloor + \left \lfloor a \right \rfloor +1  = 2\left \lfloor a \right \rfloor +1 = \left \lfloor 2a \right \rfloor = s > s-1,$ we have that $\mathbf a \notin R_2.$
\end{proof}

The trivial generalization to $m$ variables of Proposition \ref{Prop:smallerRM-m=2} is not valid. That is, in general, it is not true that the code $\mathrm{RM}_q(s,m)$ is the smallest Reed-Muller code that contains the hyperbolic code $\mathrm{Hyp}_q(d,m),$ where $a=q-\sqrt[m]{d}$ and $s=\left \lfloor m a\right \rfloor$, check the example below.
\begin{example}
Take $q=27$, $m=3$ and $d=37$. Then $a=q-\sqrt[3]{d}=27-\sqrt[3]{37}\approx 23.66,$ and $s=\left \lfloor 3 a \right \rfloor = 71.$ It is computationally easy to check that if $i_1, i_2$ and $i_3$ are integers such that $(q-i_1)(q-i_2)(q-i_3)\geq 37,$ then $i_1+i_2+i_3 \leq 70.$ Thus $\mathrm{Hyp}_q(d,m) \subseteq \mathrm{RM}_q(s-1, m).$
\end{example}
We have the following result as the first generalization of Proposition \ref{Prop:smallerRM-m=2}.
\begin{proposition}
Given $d\in \mathbb N,$ define $a=q-\sqrt[m]{d}$ and $s=\left \lfloor ma\right \rfloor.$ Then $\mathrm{Hyp}_q(d,m)\subseteq \mathrm{RM}_q(s,m).$
Moreover, $s$ is the smallest integer with this property if $\{a\}< \frac{1}{m}$.  
\end{proposition}
\begin{proof}
Let $H \subseteq \lii 0,q-1 \rii^m$ be the set defining the hyperbolic code $\mathrm{Hyp}_q(d,m).$ We will use the following fact:
\begin{equation}
    \label{eq:optimization-Problem(m)}
    \min\left\{\sum_{j=1}^m a_j \ : \ a_j \in \mathbb{R}_{\geq 0} \hbox{ and } \prod_{j=1}^m a_j = d\right\} = m\sqrt[m]{d}.
\end{equation}
For every $\mathbf{i} = ({i}_1, \ldots, {i}_m) \in \mathbb N^m$ such that $\mathbf{i}\in H,$ we have that $\prod_{j=1}^m (q-{i}_j) \geq d$. By Equation \eqref{eq:optimization-Problem(m)}, we obtain 
$\sum_{j=1}^m (q-{i}_j) \geq m\sqrt[m]{d}$, i.e. $\sum_{j=1}^m {i}_j \leq m a.$ Since ${i}_j \in \mathbb N$ for $j\in \{1, \ldots, m\}$,
$\sum_{j=1}^m {i}_j \leq \left \lfloor m a\right \rfloor = s,$
which proves that $\mathrm{Hyp}_q(d,m) \subseteq \mathrm{RM}_q(s,m).$

Let $R \subseteq \lii 0,q-1 \rii^m$ be the set defining the code $\mathrm{RM}_q(s-1,m)$. If $\{a\}< \frac{1}{m},$ then $\left \lfloor m a\right \rfloor = m\left \lfloor a \right \rfloor$. Consider $\mathbf a = \left(\left \lfloor a\right \rfloor, \ldots, \left \lfloor a \right \rfloor \right)\in \mathbb N^m$. It is easy to check that ${\mathbf a} \in H,$ but ${\mathbf a}\notin R.$
\end{proof}
\begin{remark}
Given $d\in \mathbb N,$ take $a := q-\sqrt[m]{d}.$ Observe that $\mathrm{Hyp}_q(d,m)\not\subseteq \mathrm{RM}_q(m\lfloor a \rfloor-1, m).$ Indeed, let $H, R \subseteq \lii 0,q-1 \rii^m$ be the sets defining the codes $\mathrm{Hyp}_q(d,m)$ and $\mathrm{RM}_q(m \lfloor a \rfloor-1,m)$, respectively.  Consider $\mathbf a := \left(\left \lfloor a\right \rfloor, \ldots, \left \lfloor a \right \rfloor \right)\in \mathbb N^m$. It is easy to check that $\mathbf a \in H$, but $\mathbf a \notin R$.
\end{remark}
We now come to one of the main results of this section.
\begin{theorem}\label{Prop:smaller-m}
Given $d\in \mathbb N,$ define $a=q-\sqrt[m]{d}.$ Then $\mathrm{Hyp}_q(d,m)\subseteq \mathrm{RM}_q(s,m),$ where
\[\begin{array}{ccc}
s=m\left \lfloor a\right \rfloor + r & \hbox{ and }&
r=\left \lfloor \frac{m\log\left( \frac{q-a}{q-\left \lfloor a \right \rfloor}\right)}{\log \left( \frac{q-\left \lfloor a\right \rfloor -1}{q-\left \lfloor a \right \rfloor}\right)}\right \rfloor.
\end{array}\]
Even more, $s$ is the smallest integer with this property, that is
$$\mathrm{Hyp}_q(d,m) \not\subseteq \mathrm{RM}_q(s-1,m).$$
\end{theorem}

\begin{proof}
Let $H, R_1, R_2 \subset \lii 0,q-1 \rii^m$ be the sets defining the codes $\mathrm{Hyp}_q(d,m)$, $\mathrm{RM}_q(s,m)$ and $\mathrm{RM}_q(s-1,m)$, respectively. 
First note that, by the definition of $r$, we know that $r\in \{0,\ldots, m-1\}$ is the largest integer such that
$$(q-\left \lfloor a \right \rfloor)^{m-r} (q-\left \lfloor a \right \rfloor -1)^r \geq (q-a)^m = d.$$
Thus, if we consider 
$$\mathbf a = (\underbrace{\left \lfloor a \right \rfloor + 1, \ldots, \left \lfloor a \right \rfloor + 1}_{r}, \underbrace{\left \lfloor a \right \rfloor, \ldots, \left \lfloor a \right \rfloor}_{m-r} )\in \mathbb N^m,$$
then it is easy to check that ${\mathbf a}\in H$ but ${\mathbf a}\notin R_2$. Thus, $\mathrm{Hyp}_q(d,m) \not\subseteq \mathrm{RM}_q(s-1,m)$.

Let $\overline{R_1}$ be the complement of $R_1$ in $\lii 0,q-1 \rii^m$, {\it i.e.}
$$\overline{R_1} = \left\{{\mathbf{i}} \in \lii 0,q-1 \rii^m \ : \ \sum_{j=1}^m {i}_j \geq s+1\right\}.$$
We will show that $H\cap \overline{R_1} = \emptyset$. First note that the point $$\mathbf b = (\underbrace{\left \lfloor a \right \rfloor + 1, \ldots, \left \lfloor a \right \rfloor + 1}_{r+1}, \underbrace{\left \lfloor a \right \rfloor, \ldots, \left \lfloor a \right \rfloor}_{m-r-1} )\in \mathbb N^m,$$ satisfies that ${\mathbf b}\in \overline{R_1}$ but ${\mathbf b} \notin H$. A similar situation happens with any point obtained by a permutation of the entries of $\mathbf b$. Now let ${\mathbf i} = (i_1, \ldots, i_m) \in \overline{R_1}$ such that $\sum_{j=1}^m i_j = s+1$. If there exists an index $l$ such that $i_l > \left \lfloor a \right \rfloor +1$, there must be an index $\ell$ such that $i_\ell < \left \lfloor a \right \rfloor$. Take $\mathbf a_1 = \mathbf i - \mathbf e_l + \mathbf e_\ell,$ where $\mathbf e_i$ denotes the $i$-th standard vector in $\mathbb N^m.$ Define the function 
\[f(X_1, \ldots, X_m) = \prod_{j=1}^m (q-X_j).\]
It is easy to check that $f(\mathbf i) < f(\mathbf a_1)$. Indeed,
\begin{eqnarray*}
f(\mathbf i) < f(\mathbf a_1) & \Longleftrightarrow &
(q-i_l) (q-i_\ell) < (q-a_{1,l})(q-a_{1,\ell})\\
& \Longleftrightarrow & i_{\ell} < i_{l}-1.
\end{eqnarray*}
Now, if there exists again an index $l_2$ such that $a_{1,l_2}> \left \lfloor a \right \rfloor +1$, then there must exists an index $\ell_2$ such that $a_{1,\ell_2} < \left \lfloor a \right \rfloor$. Then we can repeat the process until we reach a permutation of the entries of the point $\mathbf b$. Thus, we get a set of points $\{\mathbf a_i\}_{i=1, \ldots, s}$ such that $f(\mathbf i) < f(\mathbf a_1) < \ldots < f(\mathbf a_s) < f(\mathbf b) < d$. That is ${\mathbf b}\in \overline{R_1},$ but ${\mathbf b}\notin H$.
\end{proof}
\begin{example}
The Reed-Muller $\mathrm{RM}_9(s,2)$ are hyperbolic codes for $s\leq 4$ and $s\geq 14$ by Theorem \ref{20.08.16}. 
\end{example}
We close this section with an example that shows the smallest Reed-Muller that contains a hyperbolic code.
\begin{example}\label{22.11.25}
The lattice points under the red curve of Figure~\ref{fig:sub1} define the hyperbolic code $\mathrm{Hyp}_9(27,2)$. By Theorem~\ref{Prop:smaller-m}, we have that $\mathrm{Hyp}_9(27,2)\subseteq \mathrm{RM}_9(s,2)$ when $s\geq 7$. The lattice points under the blue curve of Figure~\ref{fig:sub1} define the Reed-Muller hyperbolic code $\mathrm{RM}_9(7,2)$, which is the smallest Reed-Muller code that contains $\mathrm{Hyp}_9(27,2)$.
\end{example}

\begin{example}\label{22.11.26}
The lattice points under the red curve of Figure~\ref{fig:sub2} define the hyperbolic code $\mathrm{Hyp}_9(9,2)$. By Theorem~\ref{Prop:smaller-m}, we have that $\mathrm{Hyp}_9(9,2)\subseteq \mathrm{RM}_9(s,2)$ when $s\geq 12$. The lattice points under the blue curve of Figure~\ref{fig:sub2} define the Reed-Muller hyperbolic code $\mathrm{RM}_9(12,2)$, which is the smallest Reed-Muller code that contains $\mathrm{Hyp}_9(9,2)$.
\end{example}

\begin{figure}[h]
\centering
\begin{subfigure}{.4\textwidth}
\centering
  \includegraphics[scale=0.6]{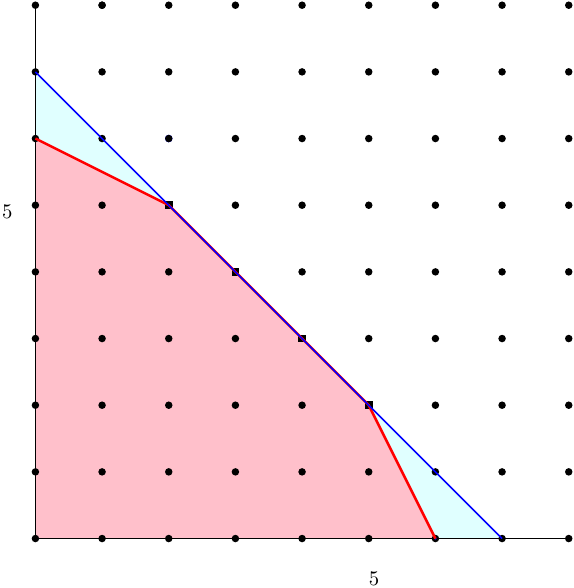}
  \caption{}
  \label{fig:sub1}
  \end{subfigure}
  \hspace{1cm}
\begin{subfigure}{.4\textwidth}
\centering
    \includegraphics[scale=0.6]{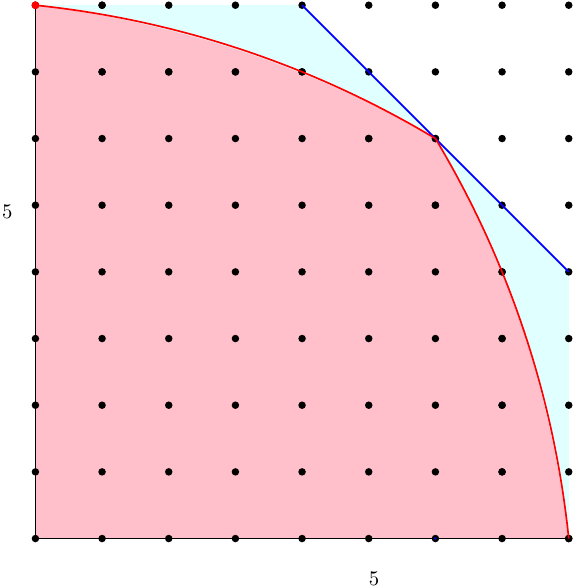}
    \caption{}
     \label{fig:sub2}
  \end{subfigure}
  \caption{(A) The lattice points under the red curve define $\mathrm{Hyp}_9(27,2)$. The lattice points under the blue curve define $\mathrm{RM}_9(7,2)$, the smallest Reed-Muller code that contains $\mathrm{Hyp}_9(27,2)$.
(B) The lattice points under the red curve define $\mathrm{Hyp}_9(9,2)$. The lattice points under the blue curve define $\mathrm{RM}_9(12,2)$, the smallest Reed-Muller code that contains $\mathrm{Hyp}_9(9,2)$.}
\end{figure}

\section{The largest Reed-Muller code} \label{secC4}
Given the hyperbolic code over $\mathbb F_q$ of degree $d$ with $m$ variables $\mathrm{Hyp}_q(d,m)$,
we now find the largest degree $s$ such that $\mathrm{RM}_q(s,m)\subseteq \mathrm{Hyp}_q(d,m).$
We first recall the dimension of a Reed-Muller code.
\begin{proposition} (\cite{G08})\label{21.04.02}
Take $s\leq (q-1)m$. Write $s=t(q-1)+r,$ where $t,r \in \mathbb{N}$ and $0\leq r < q-1.$ The minimum distance of the Reed-Muller code $\mathrm{RM}_q(s,m)$ is $\delta = (q-r)q^{m-1-t}$.
\end{proposition}

\begin{proposition}\label{22.11.33}
Let $d \in \mathbb{Z}_{\geq 1}.$ The minimum distance $\delta\left(\mathrm{RM}_q(s,m)\right) \geq d$ if and only if $$s \leq (m - c)(q-1) + q - \left\lceil \frac{d}{q^{c-1}} \right\rceil\hbox{, where }c := \left\lceil \log_q(d) \right\rceil.$$
\end{proposition}
\begin{proof}
Write $s = t (q-1) + r,$ where $t , r \in \mathbb{N}$ and $0 \leq r < q-1$. By Proposition~\ref{21.04.02},
\[ \delta \left(\mathrm{RM}_q(s,m)\right) = (q-r) q^{m-1-t}.\]

$(\Rightarrow)$ Assume that  $\delta(\mathrm{RM}_q(s,m)) \geq d.$ Then we have that $q^{m-t} \geq  (q-r) q^{m-1-t} \geq d.$ Because the logarithmic properties, $m-t \geq \left\lceil \log_q(d) \right\rceil$ and $t \leq m - \left\lceil \log_q(d) \right\rceil = m - c$. For $t < m-c$ we have that $s < (t+1)(q-1) \leq (m-c)(q-1)$ and the result follows.  Moreover, for $t = m - c,$ then $q-r \geq \left\lceil d/q^{m-1-t} \right\rceil = \left\lceil d/q^{c - 1} \right\rceil$ and, hence, $r \leq  q - \left\lceil \frac{d}{q^{c-1}} \right\rceil$. 
Putting all together we have that $s \leq (m - c)(q-1) + q - \left\lceil \frac{d}{q^{c-1}} \right\rceil.$

$(\Leftarrow)$ Conversely, if $s < (m - c)(q-1) + q - \left\lceil \frac{d}{q^{c-1}} \right\rceil,$ then $\delta(\mathrm{RM}_q(s,m)) \geq d$.
\end{proof}

We come to one of the main results of this section, which helps to find the largest Reed-Muller code inside of a hyperbolic code.
\begin{theorem}
\label{cor:highest}
Let $d \in \mathbb{Z}_{\geq 1}$. Then $\mathrm{RM}_q(s,m) \subseteq \mathrm{Hyp}_q(d,m)$ if and only if 
\[s \leq (m - c)(q-1) + q - \left\lceil \frac{d}{q^{c-1}} \right\rceil\hbox{, where }c := \left\lceil \log_q(d) \right\rceil.\]
\end{theorem}
\begin{proof}
This is a direct consequence of Proposition~\ref{22.11.33}.
\end{proof}
\begin{example}\label{22.11.28}
The lattice points under the red curve of Figure~\ref{fig:sub5} define the hyperbolic code $\mathrm{Hyp}_9(27,2)$. By Theorem~\ref{cor:highest}, we have that $\mathrm{RM}_9(s,2) \subseteq \mathrm{Hyp}_9(27,2) $ when $s\leq 6$. The lattice points under the black curve of Figure~\ref{fig:sub5} define the Reed-Muller hyperbolic code $\mathrm{RM}_9(6,2)$, which is the largest Reed-Muller in $\mathrm{Hyp}_9(27,2)$.
\end{example}

\begin{example}\label{22.11.29}
The lattice points under the red curve of Figure~\ref{fig:sub6} define the hyperbolic code $\mathrm{Hyp}_9(9,2)$. By Theorem~\ref{cor:highest}, we have that $\mathrm{RM}_9(s,2) \subseteq \mathrm{Hyp}_9(9,2)$ when $s\leq 8$. The lattice points under the black curve of Figure~\ref{fig:sub6} define the Reed-Muller hyperbolic code $\mathrm{RM}_9(8,2)$, which is the largest Reed-Muller in $\mathrm{Hyp}_9(9,2)$.
\end{example}

\begin{figure}[h]
\centering
\begin{subfigure}{.4\textwidth}
\centering
  \includegraphics[scale=0.6]{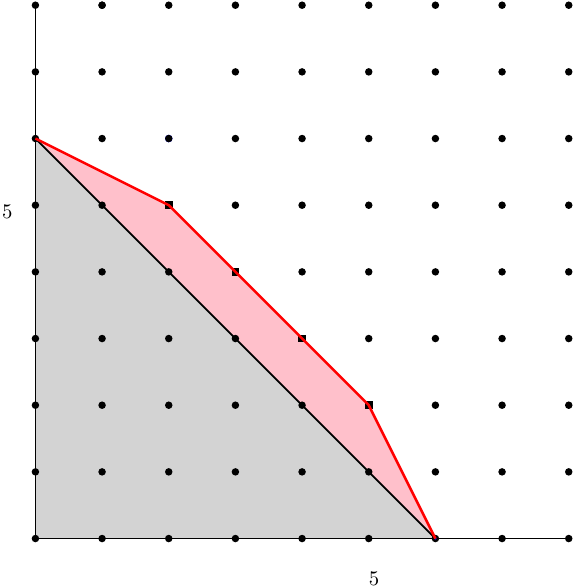}
  \caption{}
  \label{fig:sub5}
  \end{subfigure}
  \hspace{1cm}
\begin{subfigure}{.4\textwidth}
\centering
    \includegraphics[scale=0.6]{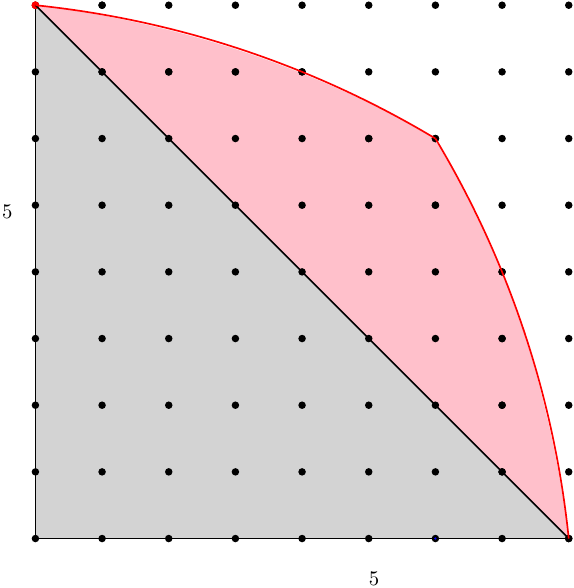}
    \caption{}
     \label{fig:sub6}
  \end{subfigure}
  \caption{(A) The lattice points under the red curve define $\mathrm{Hyp}_9(27,2)$. The lattice points under the black curve define $\mathrm{RM}_9(6,2)$, the largest Reed-Muller code in $\mathrm{Hyp}_9(27,2)$. 
(B) The lattice points under the red curve define $\mathrm{Hyp}_9(9,2)$. The lattice points under the black curve define $\mathrm{RM}_9(8,2)$, the largest Reed-Muller in $\mathrm{Hyp}_9(9,2)$.}
\end{figure}

\section{Generalized Hamming weights}\label{21.04.04}
This section proves that, similarly to Reed-Muller and Cartesian codes, the $r$-th generalized Hamming weight and the $r$-th footprint of the hyperbolic code coincide. Unlike Reed-Muller~\cite{HeiPel} and Cartesian~\cite{BEELEN2018130}, determining the $r$-th footprint of a hyperbolic code is still an open problem. We give upper and lower bounds for the $r$-th footprint of a hyperbolic code $\mathrm{Hyp}_q(d,m)$ in terms of the largest Reed-Muller code $\mathrm{RM}_q(s,m)$ contained in $\mathrm{Hyp}_q(d,m)$ and the smallest Reed-Muller code $\mathrm{RM}_q(s^\prime,m)$ that contains $\mathrm{Hyp}_q(d,m)$. These bounds sometimes are sharp.

Recall that the monomial code associated with $A \subseteq \lii 0, q-1\rii ^m$ is given by
\[\mathcal C_A = \mathrm{ev}_{\mathcal P} (\mathbb F_q[A]) = \left\{\mathrm{ev}_{\mathcal P}(f) \ : \ f\in \mathbb F_q[A] \right\}.\]
For and an integer $1 \leq r \leq |A|$, the $r$-th generalized Hamming weight of $\mathcal{C}_A$ is given by
\[\delta_r(\mathcal C_A) = \min \{ \ |\chi(\mathcal D)| \ : \ \mathcal D \subseteq \mathcal C_A, \ \dim (\mathcal D) = r \},\]
where $\chi(\mathcal D) := \left\{i \in [n] \ | \ \text{there is } \mathbf x \in \mathcal D \text{ with } x_i \neq 0\right\}.$
We now explain how to bound the $r$-th generalized Hamming weight in terms of the footprint.
For $\mathbf i = (i_1, \ldots, i_m) \in \lii 0,q-1\rii^m$, we define the set
\[\nabla(\mathbf{i}):=\lii i_1,q-1\rii \times \cdots \times \lii i_m,q-1\rii .\]
Note that $|\nabla(\mathbf{i})|=\prod_{j=1}^m (q-i_j)$.
We can rewrite the footprint-bound of the code $C_A$ as
\[ \mathrm{FB}(C_A) =  \min \left\{ |\nabla(\mathbf{i})| \ : \ \mathbf{i} \in A \right\}.\]
The minimum distance $\delta(\mathcal C_A)$ of $\mathcal C_A$ is lower bounded by the footprint bound~\cite{GH00}:
$\mathrm{FB}(\mathcal C_A) \leq \delta(\mathcal C_A)$.
Jaramillo et al. generalized in~\cite{JVV} the footprint-bound to the {\it $r$-th footprint}:
\[ \mathrm{FB}_r(\mathcal C_A) :=
\min \left\{ \middle| \bigcup_{j=1}^r \nabla(\mathbf{i}_j) \middle| \ : \
\mathbf{i}_j \in A, \ \mathbf{i}_\ell \neq \mathbf{i}_j \text{ for } \ell, j \in \lii 1, m \rii \right\}.\]
Similarly to the minimum distance, the $r$-th generalized Hamming weight is lower bounded by the $r$-th footprint~\cite[Theorem 3.9]{JVV}:
\begin{equation}\label{22.11.27}
\mathrm{FB}_r(\mathcal C_A) \leq \delta_r(\mathcal C_A).
\end{equation}

The $r$-th footprint is sharp for Reed-Muller and Cartesian codes by~\cite{HeiPel} and~\cite{BEELEN2018130}, respectively. We now extend the result by proving that the $r$-th footprint is sharp for hyperbolic codes. Recall that the hyperbolic code $\mathrm{Hyp}_q(d,m)$ depends on the set
\[H=\left\{ \mathbf{i}  \in \lii 0,q-1\rii^m \ : \ |\nabla(\mathbf{i})| \geq d \right\}.\]
We come to one of the main results of this section.
\begin{theorem}\label{22.11.15}
The $r$-th generalized Hamming weight of a hyperbolic code $\mathrm{Hyp}_q(d,m)$ is given by the
$r$-th footprint:
\[\delta_r(\mathrm{Hyp}_q(d,m)) = \mathrm{FB}_r(\mathrm{Hyp}_q(d,m)) :=
\min \left\{ \middle| \bigcup_{j=1}^r \nabla(\mathbf{i}_j) \middle| \ : \
\mathbf{i}_j \in H, \ \mathbf{i}_\ell \neq \mathbf{i}_j \text{ for } \ell, j \in \lii 1, m \rii \right\}.\]
\end{theorem}
\begin{proof}
By Equation~(\ref{22.11.27}), $\mathrm{FB}_r(\mathrm{Hyp}_q(d,m)) \leq \delta_r(\mathrm{Hyp}_q(d,m))$.

To prove the inequality $\delta_r(\mathrm{Hyp}_q(d,m)) \leq \mathrm{FB}_r(\mathrm{Hyp}_q(d,m))$, we construct $r$ elements in $\mathrm{Hyp}_q(d,m)$ that generate a subspace in $\mathrm{Hyp}_q(d,m)$ of dimension $r$ and support length precisely $\mathrm{FB}_r(\mathrm{Hyp}_q(d,m))$. Let $\gamma$ be a primitive element of $\mathbb{F}_q$. For an nonnegative integer $\ell$, we define the polynomial $f(\ell,x)$ in $\mathbb{F}_q[x]$ of degree $\ell$ as
\[f(\ell,x):=\begin{cases}
      1 & \text{if $\ell = 0$}\\
      x & \text{if $\ell = 1$}\\
      (x)(x-\gamma)\cdots (x-\gamma^{\ell-1}) & \text{if $\ell > 1$.}
\end{cases}\]
Let $\mathbf{i}_1,\ldots,\mathbf{i}_r$ be elements in $H$ such that $\mathrm{FB}_r(\mathrm{Hyp}_q(d,m))=\left|\bigcup_{j=1}^r \nabla(\mathbf{i}_j)\right|$. For every $1\leq j\leq r$, assume $\mathbf{i}_j=(i_{j1},\ldots,i_{jm})$, and define the polynomial
\[f_j := f(i_{j1},x_1)\cdots f(i_{jm},x_m).\]
Denote by $Z(f_j)$ the set of zeros of $f_j$ in $\mathbb{F}_q^m$. Note that
\[\mathbb{F}_q^m\setminus Z(f_j) =
\left\{ \left(\gamma^{a_1},\ldots,\gamma^{a_m}\right) \in \mathbb{F}_q^m \ : \ (a_1,\ldots,a_m) \in \nabla(\mathbf{i}_j) \right\},\]
which implies that $ev_{\mathbb{F}_q^m}(f_j)\in\mathrm{Hyp}_q(d,m)$.
Let $Z(f_1,\ldots,f_r)$ be the set of common zeros of $f_1,\ldots,f_r$ in $\mathbb{F}_q^m$.
As $Z(f_1,\ldots,f_r)=\bigcap_{j=1}^r Z(f_j)$, then
$\mathbb{F}_q^m\setminus Z(f_1,\ldots,f_r) = \bigcup_{j=1}^r \left( \mathbb{F}_q^m\setminus Z(f_j)\right)$. Thus, 
if $\mathcal{D}_r:= \text{Span}_{\mathbb{F}_q}\left\{ev_{\mathbb{F}_q^m}(f_1), \ldots, ev_{\mathbb{F}_q^m}(f_r)\right\} \subseteq
\mathrm{Hyp}_q(d,m),$ then
\[|\chi( \mathcal D_r)| = |\mathbb{F}_q^m\setminus Z(f_1,\ldots,f_r)| = \left|\bigcup_{j=1}^r \nabla(\mathbf{i}_j)\right| =
 \mathrm{FB}_r(\mathrm{Hyp}_q(d,m)).\]
We conclude that $\delta_r(\mathrm{Hyp}_q(d,m))\leq\mathrm{FB}_r(\mathrm{Hyp}_q(d,m)).$
\end{proof}
We now use Theorem~\ref{22.11.15} to bound the GHWs of hyperbolic codes in terms of the $r$-th footprint.
\begin{corollary}\label{22.11.20}
Let $\mathbf{i}_1,\ldots,\mathbf{i}_r\in H$ be the first $r$ elements of $H$ in descending lexicographical order. Then
\[\delta_r(\mathrm{Hyp}_q(d,m))\leq\left|\bigcup_{j=1}^r \nabla(\mathbf{i}_j)\right|.\]
\end{corollary}
\begin{proof}
This is a direct consequence of Theorem~\ref{22.11.15}.
\end{proof}
Heijnen and Pellikaan proved in~\cite[Theorem 5.10]{HeiPel} that the bound of Corollary~\ref{22.11.20} is sharp for a Reed-Muller code. Even more, Heijnen and Pellikaan explicitly described the $r$-th generalized Hamming weight in terms of the $r$-th element in $\lii 0,q-1\rii^m$ in the lexicographic order.
Note that Theorem~\ref{22.11.15} gives an expression to compute the GHWs of a hyperbolic code in terms of finding a minimum on a set. Naturally, when the hyperbolic code coincides with a Reed-Muller code, we obtain a close formula for the GHWs of some hyperbolic codes.
\begin{theorem}
Take $m \geq 1$. Let $d$ be such that $(q-t-1)^r(q-t)^{m-r}<d,$ where $s+1=mt+r$, $0 \leq s < m(q-1)$,  and $0\leq r < m$.
The $r$-th generalized Hamming weight of the hyperbolic code $\mathrm{Hyp}_q(d,m)$ is given by:
\[\delta_r(\mathrm{Hyp}_q(d,m)) = \sum_{i=1}^ma_{m-i+1}q^{i-1}+1,\]
where $\mathbf{a}=(a_1,\ldots,a_m)$ is the $r$-th element in $\lii 0,q-1\rii^m$ in the lexicographic order with the property that $\deg(\mathbf{a}) > (q-1)m-s-1$.
\end{theorem}
\begin{proof}
By Theorem~\ref{20.08.16}, the hyperbolic code $\mathrm{Hyp}_q(d,m)$ coincides with the Reem-Muller code $\mathrm{RM}_q(s,m)$.
By~\cite[Theorem 5.10]{HeiPel}, the $r$-th generalized Hamming weight is given by
$\delta_r(\mathrm{RM}_q(s,m)) = \sum_{i=1}^ma_{m-i+1}q^{i-1}+1.$
\end{proof}

We also have the following bounds for the GHWs of an arbitrary hyperbolic code $\mathrm{Hyp}_q(d,m)$ in terms of the GHWs of some Reed-Muller codes.
\begin{corollary}
Let $\mathrm{Hyp}_q(d,m)$ be a hyperbolic code.
Define $s^\prime := m\left \lfloor a\right \rfloor + r$, where $a=q-\sqrt[m]{d}$ and
$r=\left \lfloor \frac{m\log\left( \frac{q-a}{q-\left \lfloor a \right \rfloor}\right)}{\log \left( \frac{q-\left \lfloor a\right \rfloor -1}{q-\left \lfloor a \right \rfloor}\right)}\right \rfloor$.
Let $s$ be the maximum integer such that
$s \leq (m - c)(q-1) + q - \left\lceil \frac{d}{q^{c-1}} \right\rceil$, where $c := \left\lceil \log_q(d) \right\rceil.$
Then
\[\delta_r(\mathrm{RM}_q(s^\prime,m)) \leq \delta_r(\mathrm{Hyp}_q(d,m)) \leq \delta_r(\mathrm{RM}_q(s,m)),\]
where the first inequality is valid for any $1\leq r \leq \dim ({\rm RM}_q(s,m))$ and
the second inequality is true for any $1\leq r \leq \dim (\mathrm{Hyp}_q(d,m))$.
\end{corollary}
\begin{proof}
The result follows from Theorems~\ref{Prop:smaller-m} and~\ref{cor:highest}, where we prove
$\mathrm{RM}_q(s,m) \subseteq \mathrm{Hyp}_q(d,m)\subseteq \mathrm{RM}_q(s^\prime,m).$
\end{proof}

\begin{example}\label{22.11.24}
From Examples~\ref{22.11.25} and~\ref{22.11.28}, we have that
$\mathrm{RM}_9(6,2)\subseteq\mathrm{Hyp}_9(27,2)\subseteq \mathrm{RM}_9(7,2)$.
 Thus, 
\[ 26 = \delta_2(\mathrm{RM}_9(7,2)) \leq \delta_2(\mathrm{Hyp}_9(27,2)) \leq \delta_2(\mathrm{RM}_9(6,2)) = 35.\] 
Using computational software and Theorem~\ref{22.11.15}, we can see that the actual value is
$\delta_2(\mathrm{Hyp}_9(27,2)) =  {\rm FB}_2(\mathrm{Hyp}_9(27,2)) = 32$ (see Figure \ref{fig:ghw1}).
\end{example}
\begin{figure}[h]
\centering
\begin{subfigure}{.4\textwidth}
\centering
  \includegraphics[scale=0.6]{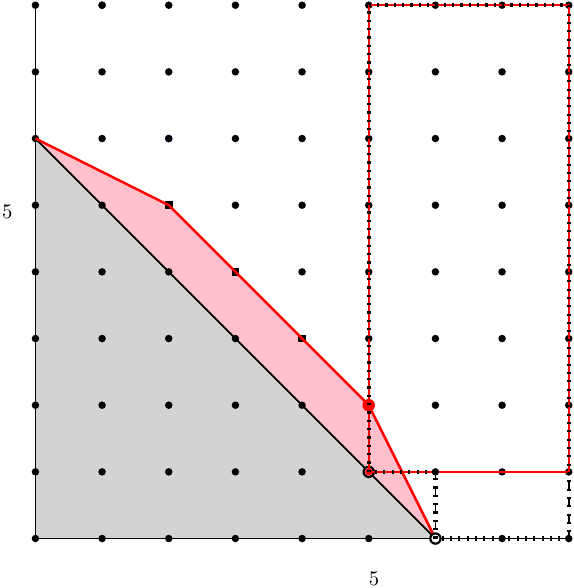}
  \caption{}
  \label{fig:sub3}
  \end{subfigure}
  \hspace{1cm}
\begin{subfigure}{.4\textwidth}
\centering
    \includegraphics[scale=0.6]{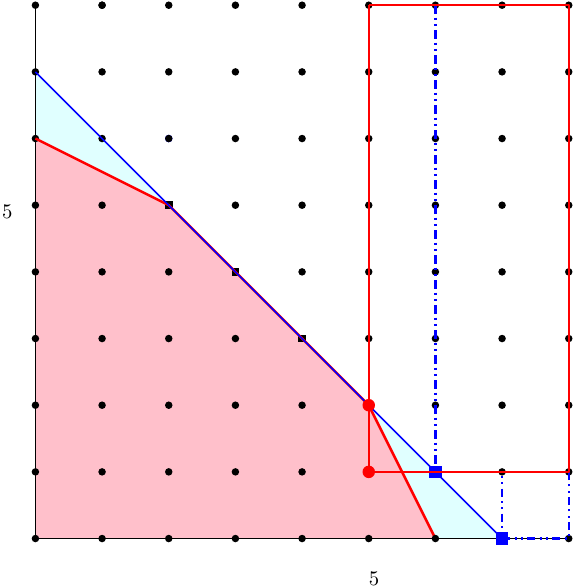}
    \caption{}
  \end{subfigure}
  \caption{We observe that $\mathrm{RM}_9(6,2)\subseteq\mathrm{Hyp}_9(27,2)\subseteq \mathrm{RM}_9(7,2)$. The boxes represent the lattice points that help to compute the second GHWs. The number of lattice points inside: the red box is equal to $\delta_2(\mathrm{Hyp}_9(27,2))$, the black box is equal to $\delta_2(\mathrm{RM}_9(6,2))$, and the blue box is equal to $\delta_2(\mathrm{RM}_9(7,2))$.}
      \label{fig:ghw1}
\end{figure}

\begin{example}
From Examples~\ref{22.11.26} and~\ref{22.11.29}, we have that
$\mathrm{RM}_9(8,2)\subseteq\mathrm{Hyp}_9(9,2)\subseteq \mathrm{RM}_9(12,2)$.
 Thus, 
\[ 9 = \delta_2(\mathrm{RM}_9(12,2)) \leq \delta_2(\mathrm{Hyp}_9(9,2)) \leq \delta_2(\mathrm{RM}_9(8,2)) = 17.\] 
Using computational software and Theorem~\ref{22.11.15}, we can see that the actual value is
$\delta_2(\mathrm{Hyp}_9(9,2)) =  {\rm FB}_2(\mathrm{Hyp}_9(9,2)) = 12$ (see Figure \ref{fig:ghw2}).
\end{example}
\begin{figure}[h]
\centering
\begin{subfigure}{.4\textwidth}
\centering
  \includegraphics[scale=0.6]{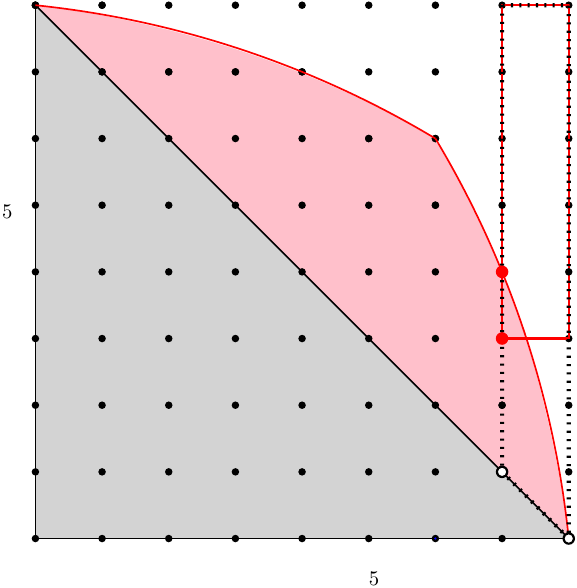}
  \caption{}
  \label{fig:sub3}
  \end{subfigure}
  \hspace{1cm}
\begin{subfigure}{.4\textwidth}
\centering
    \includegraphics[scale=0.6]{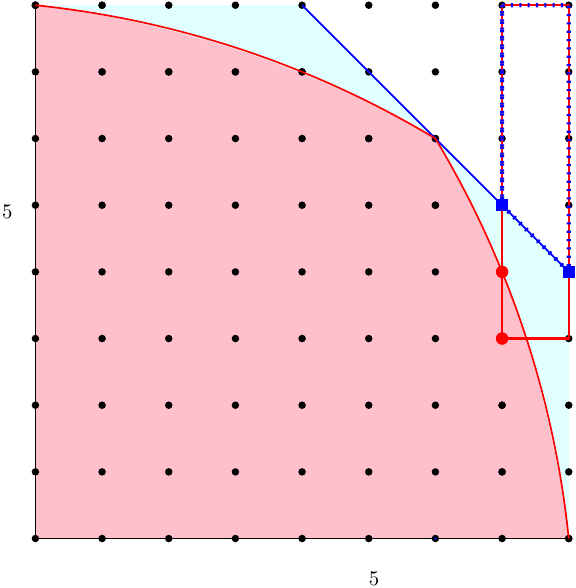}
    \caption{}
  \end{subfigure}
  \caption{We observe that $\mathrm{RM}_9(8,2)\subseteq\mathrm{Hyp}_9(9,2)\subseteq \mathrm{RM}_9(12,2)$. The boxes represent the lattice points that help to compute the second GHWs. The number of lattice points inside: the red box is equal to $\delta_2(\mathrm{Hyp}_9(9,2))$, the black box is equal to $\delta_2(\mathrm{RM}_9(8,2))$, and the blue box is equal to $\delta_2(\mathrm{RM}_9(12,2))$.}
      \label{fig:ghw2}
\end{figure}
\label{example-smallest}

The following example shows that the bounds of Corollary~\ref{22.11.20} may be sharp for some of the GHWs of a hyperbolic code.
\begin{example}
Let $q=9$ and $H=\{(i_1,i_2)\in \lii 0,8 \rii \ |\ (9-i_1)(9-i_2)\geq 27\}$.
By computational software, the first generalized Hamming weight, which is the minimum distance, is given by the first element of $H$ in descending lexicographical order. This first element is $(6,0)$. See Figure~\ref{22.11.30}. As $|\nabla(6,0)|=27$, we obtain $\delta_1(\mathrm{Hyp}_9(27,2))=27$. 
\begin{figure}[h]
\includegraphics[width=.3\textwidth]{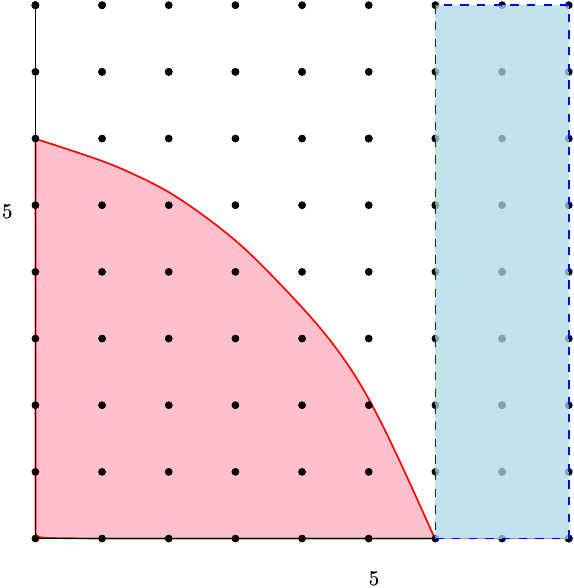}
\caption{The number of lattice points inside of the blue box equals $\delta_1(\mathrm{Hyp}_9(27,2))$.}
\label{22.11.30}
\end{figure}

The first two elements in descending lexicographical order in $H$ are $(6,0)$ and $(5,2)$. See Figure~\ref{22.11.31}.
We obtain $|\nabla(6,0)\cup\nabla(5,2)|=34$. However, $\delta_2(\mathrm{Hyp}_9(27,2)) = 32$ by Example~\ref{22.11.24}, which means that the first two elements do not give the second weight in descending lexicographical order.
\begin{figure}[h]
\includegraphics[width=.3\textwidth]{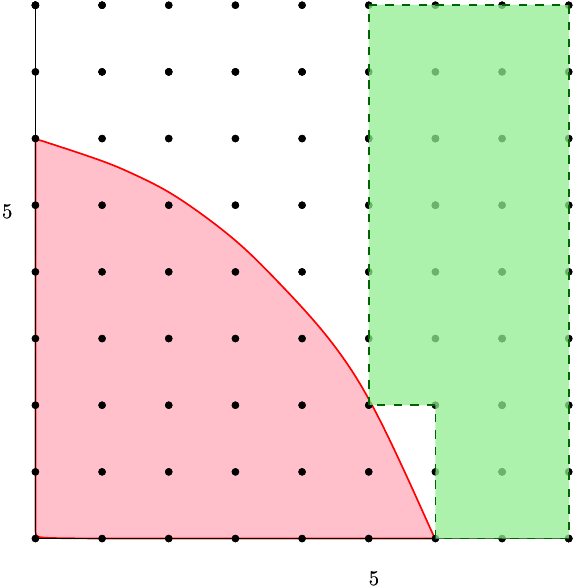}
\caption{The number of lattice points inside of the green box equals $\delta_2(\mathrm{Hyp}_9(27,2))$.}
\label{22.11.31}
\end{figure}

The first fourth elements in descending lexicographical order in $H$ are $(6,0), (5,2), (5,1),$ and $(5,0)$. See Figure~\ref{22.11.32}.
The first and fourth elements, respectively, give the third and fourth GHWs:
\[\delta_3(\mathrm{Hyp}_9(2,27))=|\nabla(6,0)\cup\nabla(5,2)\cup\nabla(5,1)|=35\]
Thus,
\[\delta_4(\mathrm{Hyp}_9(2,27))=|\nabla(6,0)\cup\nabla(5,2)\cup\nabla(5,1)\cup\nabla(5,0)|=36.\]
\begin{figure}[h]
\includegraphics[width=.3\textwidth]{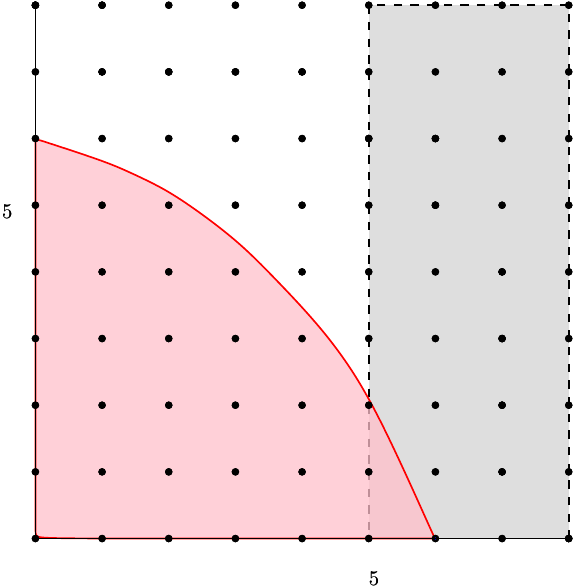}
\caption{The number of lattice points inside of the grey box equals
$\delta_4(\mathrm{Hyp}_9(2,27))=|\nabla(6,0)\cup\nabla(5,2)\cup\nabla(5,1)\cup\nabla(5,0)|=36$.}
\label{22.11.32}
\end{figure}
\end{example}

\section*{Acknowledgments} 
Part of this work was developed while E. Camps-Moreno visited Universidad de La Laguna.

\bibliography{codes}{}
\bibliographystyle{abbrv}

\end{document}